\documentclass[11pt]{article}

\UseRawInputEncoding 

\newtheorem{theorem}{Theorem}[section]

\newtheorem{conjecture}[theorem]{Conjecture}

\newcommand\qed{\begin{flushright} {\bf q.e.d.} \end{flushright} }
\newcommand\prf{\noindent {\bf Proof :}}  

\newcommand\bits{\{0,1\}}
\newcommand\uu{{\bits^*}}
\newcommand\nn{{\bits^n}}
\newcommand\NN{{\bits^N}}
\newcommand\mm{{\bits^m}}
\newcommand\kk{{\bits^k}}
\newcommand\nnn{{\bits^{n+1}}}
\newcommand\NNN{{\bits^{N+1}}}

\newcommand\nlog{{{n+\lceil\log n\rceil + 1}}}
\newcommand\nl{{{n+\lceil\log n\rceil}}}

\newcommand\pp{{\cal P}}
\newcommand\np{{\cal N}{\cal P}}
\newcommand\npco{{\np \cap \conp}}
\newcommand\conp{{\mbox{co}\np}}
\newcommand\ee{{\cal E}}

\newcommand\gadf{{\mbox{Gad}_{f}}}
\newcommand\sz{{\mbox{Size}}}
\newcommand\prob{{\mbox Prob}}
\newcommand\bs{{\bf s}}

\newcommand\taut{\mbox{TAUT}}
\newcommand\tru{{\mbox{{\bf tt}}}}

\begin{document}

\title{Failure of the strong feasible disjunction property} 

\author{Jan Kraj\'{\i}\v{c}ek}

\date{Faculty of Mathematics and Physics\\
Charles University\thanks{Sokolovsk\' a 83, Prague, 186 75,
The Czech Republic, {\tt jan.krajicek@protonmail.com}}}

\maketitle

\begin{abstract}

A propositional proof system $P$ has the strong feasible disjunction property
iff there is a constant $c \geq 1$ such that whenever $P$ admits a size $s$
proof of $\bigvee_i \alpha_i$ with no two $\alpha_i$ sharing an atom then
one of $\alpha_i$ has a $P$-proof of size $\le s^c$.

We combine the work of Ilango \cite{Ila} and Ren et al. \cite{RWZ} 
with the gadget proof complexity generator of \cite{Kra-generator} and 
rule out the property for strong enough proof systems under the following 
two hypotheses:

\begin{itemize}

\item there exists a language $L \in \ee$ that requires size $2^{\Omega(n)}$
circuits even if they are allowed to query an $\np$ oracle,

\item there exists a $\pp/poly$ generator $G : \nn \rightarrow \nnn$
which is a demi-bit in the sense of Rudich \cite{Rud}.

\end{itemize}

\end{abstract}


\section*{Introduction}

A {\bf propositional proof system} (abbr. {\bf pps}) is a p-time map 
$P : \uu \rightarrow \uu$ whose range $Rng(P)$ is exactly the set $\taut$ of 
propositional tautologies, cf. Cook and Reckhow \cite{CooRec}.
We denote by the dot above the disjunction sign:
\begin{equation} \label{28.3.26a}
\dot{{\bigvee}}_{1\le i \le r} \alpha_i
\end{equation}
disjunctions in which  no two formulas $\alpha_i$ have an atom in common.
A pps $P$ has the 
{\bf strong\footnote{The qualification {\em strong } refers to the fact 
that we allow any $r$.} feasible disjunction property} (abbreviated {\bf strong fdp}) iff 
there exists a constant $c \geq 1$ such that whenever a disjunction (\ref{28.3.26a})
has a $P$-proof of size $s$
then one of $\alpha_i$ has a $P$-proof of size $\le s^c$. 

The fdp without the qualification strong (i.e. the case $r=2$) was investigated since 1980s. 
The stronger property was studied in connections with the proof complexity of the 
Nisan-Wigderson generator, cf. \cite{Kra-nwg}.  

There is also a connection to feasible interpolation (cf. \cite{prf}).
Razborov \cite{Raz95} defined a first-order formal system with 
in-built interpolation and linked the provability of circuit lower bounds in 
that system to the existence of $\pp/poly$ natural proofs against $\pp/poly$. 
A more fundamental level at which to study the provability of circuit lower bounds
is propositional logic (cf. \cite[Sec.19.5]{prf}) and Razborov's argument can be 
turned into a statement
refuting feasible interpolation for strong enough propositional proof systems
assuming that some explicit language requires large circuits 
and that $\pp/poly$ natural proofs against $\pp/poly$ do not exist. 
An entirely analogous argument rules out 
the fdp when assuming that even $\np/poly$ natural proofs against $\pp/poly$ do not
exist, cf. Rudich \cite[Thm.1]{Rud}.
On the other hand, if a proof system admits feasible interpolation then it also 
admits fdp, and when the bound in feasible interpolation is linear then
it implies (by repeated application) even the strong fdp. Proof systems admitting
feasible interpolation are discussed in \cite[Chpts.17 ad 18]{prf} and the strong fdp
is discussed (with further references) in \cite[Subsec.17.9.2]{prf}.

A failure of the strong fdp for strong proof systems (to be defined in Sec.\ref{prelim})
was previously also established using some
assumptions from proof complexity: Khaniki \cite{Kha23} used an assumption about
the provability of 
reflection principles for an implicit proof systems and \cite{k4} used the theory of 
proof complexity generators.

In this note we apply the idea behind \cite[Thm.4.2.7]{k4}
and a recent work
of Ilango \cite{Ila} and Ren et al. \cite{RWZ} and we derive the failure of the strong fdp
from the following two hypotheses (both defined in Sec.\ref{prelim}): 
that some language in the class $\ee$ 
(deterministic time $2^{O(n)}$)
does not have sub-exponential size circuits even 
if they are allowed to query an $\np$ oracle and the demi-bit conjecture
of Rudich \cite{Rud}.
Our main tool is the gadget generator of \cite{Kra-generator} and 
its property that it can, in a specific sense, hide the non-uniformity of other
generators. 

It would be of interest (to this author at least) to deduce the failure of the
strong fdp from some reasonable hypotheses that themselves do not postulate 
some proof complexity lower bounds. The hypothesis (E-NP) fulfills this requirement
but the demi-bit conjecture does not.
But the advantage of the demi-bit conjecture
over the hypothesis about the non-existence of 
$\np/poly$ natural proofs against $\pp/poly$ is that it is a weaker assumption
and does not target the very specific lower bound tautologies.

\medskip

The paper is organized as follows. Sec.\ref{prelim} gives some 
proof complexity preliminaries.
In Sec.\ref{demi} we present Ilango's \cite{Ila} (somewhat simplified)
construction of a generator from demi-bits.
The main theorem is stated and proved in Sec.\ref{thm}. 
The paper is concluded by a discussion in Sec.\ref{unif} 
of the possibility to construct a uniform hard generator.

We give original references for notions and results in 
the theory of proof complexity generators but all of that can be found also 
in \cite{k4}. The reader requiring a proof complexity background can find it in \cite{prf}.

\section{Proof complexity preliminaries}  \label{prelim}
  
A {\bf propositional proof system} (abbr. {\bf pps}) is a p-time map $P : \uu \rightarrow \uu$
whose range $Rng(P)$ is exactly the set $\taut$ of propositional tautologies, 
cf. Cook and Reckhow \cite{CooRec}.
A pps with {\bf $k(n)$ bits of advice}\footnote{Note that pps can be defined equivalently 
using the provability relation but for pps with advice such definitions are
not equivalent, cf. \cite{CK}.}
(abbr. pps$/k(n)$) is such $P$ computed by a p-time algorithm with $k(n)$ bits of advice on inputs of size $n$, cf. Cook and K. \cite{CK}.
  
A pps $P$ is {\bf strong} iff $P$ is Extended Frege system EF augmented by a p-time subset $A\subseteq \taut$ as additional axioms:
any substitution instance of any formula in $A$ can be used in a proof.
Strong proof systems have useful technical properties; for example, they simulate the modus ponens rule and substitution of constant in p-size.
Any pps can be p-simulated by a strong proof system. See \cite[Sec.2.4]{k4}.
  
The {\bf lengths-of-proofs function} $\bs_P$ is defined as:
$$
\bs_P(\varphi)\ :=\ \min\{|\pi|\ |\ P(\pi)=\varphi    \}\ .
$$
A pps $P$ {\bf simulates} pps $Q$ (notation $P \geq Q$) iff 
$\bs_P(\varphi)\le \bs_Q(\varphi)^{O(1)}$ and $P$ {\bf p-simulates} $Q$ 
($P \geq_p Q$) iff there is a p-time map $f$ such that $P(f(\pi)) = Q(\pi)$. 
  
\begin{theorem}[{Cook-K.\cite[Thm.6.6]{CK}}] \label{ck}
There exists a pps $P$ with $1$ of advice that p-simulates all classical Cook-Reckhow
pps.
\end{theorem}
We remark that $P$ also simulates all other pps with $O(\log n)$ bit of advice
but those simulations are non-uniform.
  
\bigskip
  
A {\bf generator} is any map $g : \uu \rightarrow \uu$ such that its restriction $g_n$ 
to $\nn$ maps $\nn$ into $\mm$ where $m = m(n) > n$ ($g$ is {\bf stretching})
and is computed\footnote{Note that maps computed in (non-uniform)
$\np \cap co\np$ are also useful in this context, cf. \cite{k4}.}
by a circuit $C_n$ of size polynomial in $m$.
The {\bf $\tau$-formula} $\tau(g)_b$ for $b \in \mm$ is a canonical propositional formula 
expressing that $b \notin Rng(C_n)$. cf. \cite{Kra-wphp,ABRW}. 
Its size is polynomial in $m$.
  
An example of a generator related to the discussion in the Introduction 
is the truth-table function $\tru_{s,k}$, $s=s(k)$
and $s \log s \ll 2^k$, sending $O(s \log s)$ bits describing a size $s$ circuit 
with $k$ inputs into the size $2^k$ truth table of the function it computes.
The formula $\tau(\tru_{s,k})_b$ then expresses that the Boolean function with 
the table $b$ has the circuit size bigger than $s$.

A generator $g$ is {\bf hard for $P$} iff for any $c \geq 1$ the inequality
$$
\bs_P(\tau(g)_b) \ \le\ |\tau(g)_b|^c\ \le \ |b|^{O(c)}
$$
holds for only a finite number of $b \in \uu$.
  
A hypothesis motivating a significant part of the theory of proof complexity generators
is the following one.

\begin{conjecture}[{\cite{Kra-dual}, \cite{Kra-strong}, \cite[Sec4.2]{k4}}] \label{conj}
There exists a generator 
hard for all pps (equivalently, its range intersects all 
infinite $\np$ sets).

In fact, such $g$ exists p-time computable.

\end{conjecture}

Our tool will be an argument underlying Theorem \ref{enp} that uses the following 
hypothesis:

\begin{itemize}
	
	\item [(E-NP)] 
	{\em There exists a language $L \in \ee$ such that for every $A \in \np$
		there is $\epsilon > 0$
		such that $L \notin_{i.o.} Size^A(2^{\epsilon n})$, where
		$\sz^A(s(n))$ is the class of languages $L$ such that
		$L_n$, all $n \geq 1$, can be computed by a circuit of size $\le s(n)$
		that is allowed to query oracle $A$.}
\end{itemize}
The notation for the hypothesis is the same as in \cite{k4}.

\begin{theorem}[{\cite[Thm.4.2.7]{k4}}] \label{enp}
{\ }

Assume hypothesis (E-NP) and
that Conjecture \ref{conj}
holds true for a p-time generator $g$.	
Then there exists a proof system $Q$ such that no strong proof system $P$ 
that simulates $Q$ has the fdp. 
	
\end{theorem}

\section{Generator from demi-bits}  \label{demi}

Rudich \cite{Rud} studied the possibility to generalize natural proofs 
to $\np$-natural proofs and for that he proposed
two strengthenings of the pseudo-random generator hypothesis to the statements 
that there is
a $\pp/poly$ generator whose range intersects all $\np$ sets with a 
non-subexponential density
(the demi-bit conjecture) or even having the intersection of a similar 
relative density in the range as is the density of the $\np$ set (the super-bit conjecture). 
We shall define the former bellow.

A non-deterministic circuit $D$ with $m$ inputs defines a subset of $\mm$: all strings that
are accepted by $D$ (i.e. the circuit evaluates to $1$ for some choice of the non-deterministic
input bits).
Given a generator
$g$, $g_n : \nn \rightarrow \mm$, its {\bf demi-hardness}
is a function $s(n)$ such that $s(n)$ is the minimal size of a non-deterministic
circuit $D$ that defines a subset of $\mm \setminus Rng(g_n)$ 
with measure in $\mm$ at least $s(n)^{-1}$.

\begin{conjecture}[{Demi-bit conjecture, Rudich \cite[Conj.5]{Rud}}] \label{rud}
	{\ }
	
There exists $\pp/poly$ generator $g$ with the stretch $m(n) = n + 1$ having 
the demi-hardness $2^{n^{\epsilon}}$ for some $\epsilon > 0$.
\end{conjecture}
In fact, Rudich \cite{Rud} proposes a uniform p-time candidate
based on subset sum, cf. \cite[Conj. 4]{Rud}. 
The stretch can be improved too but not too much: Tzameret and Zhang
\cite{TzaZha} showed how to get the stretch $n + n^{1-\Omega(1)}$ (we will need only
$\nlog$).

Assuming the demi-bit conjecture Ilango \cite{Ila} and Ren\footnote{Their construction
needs a larger stretch $n^{O(1)}$ that is not 
known to follow from Conjecture \ref{rud}.} 
et al. \cite{RWZ} proved
a weaker form of Conjecture \ref{conj}: 
for every pps $P$ there is a $\pp/poly$ proof complexity
generator $g_P$ hard for $P$.
Both papers added something extra. The generator $g_P$ 
constructed in Ren et al. \cite{RWZ}
is, in fact, pseudo-surjective for $P$ if the number of rounds in the definition of pseudo-surjectivity is suitably limited
(we shall not define here the pseudo-surjectivity, cf. \cite{Kra-dual}, but we will 
return to this in Sec. \ref{unif}). The extra in 
\cite{Ila} is the following theorem.

\begin{theorem} [Ilango \cite{Ila}] \label{ila}
Assuming the demi-bit conjecture there is a $\pp/poly$ 
proof complexity generator $g$ hard for all pps $P$.
\end{theorem}

The construction of $g_P$ hard for $P$   
in \cite{Ila} (especially as presented in \cite[Appendix A]{RWZ}) is very
simple and the theorem was deduced from it using the countability of 
the class of Cook-Reckhow proof systems and the Borel-Cantelli lemma. 
We shall present a somewhat simplified proof using proof systems with 
advice\footnote{\cite{RWZ} formulate their results for non-uniform proof systems 
but what they actually mean are $\np/poly$ subsets
of $\taut$.}
and Theorem \ref{ck}.

\bigskip
\noindent
{\bf Proof of Theorem \ref{ila}:}

\paragraph {(1)} Let $m := \nlog$ and take $G : \nn \rightarrow \mm$ a generator with the
demi-bit hardness $s(n) \geq n^{\omega(1)}$ guaranteed to exists by Conjecture \ref{rud}.

\paragraph {(2)} Let $P$ be a pps/1 simulating all Cook-Reckhow proof systems provided by
Theorem \ref{ck}.

\paragraph {(3)} Take any $t(n)$ such that $n^{\omega(1)} \le t(n) \le s(n)^{o(1)}$
and define $A \subseteq \mm$ as
$A\ :=\ \{ a \in \mm\ |\ \bs_P(\tau(G)_a) > t(n) \}$.
Then $\mm \setminus A$ is disjoint with $Rng(G)$ and it is
defined by a non-deterministic circuits of size 
$t(n)^{O(1)} \le s(n)-1$ and hence the measure of $\mm \setminus A$
must be less than $1/(s(n)-1)$. This implies that the measure of
$A$ satisfies $|A|2^{-m} > 1 - 1/(s(n)-1)$.

\paragraph {(4)} {\bf Claim} (Lautemann \cite{Lau83}):
$$
\prob_{s_1,\dots, s_n \in \mm}[\mm = \bigcup_j s_j \oplus A] \geq 
1- (s(n)-1)^{-n}2^m \geq 1 - o(1)\ .
$$

\paragraph {(5)} Following Ilango \cite{Ila} as presented in Ren et al.
\cite[Appendix A]{RWZ} define, given $s = (s_1, \dots, s_n) \in (\mm)^n$,
generator
$$
g_s : \bits^\nl \rightarrow \bits^\nlog
$$
that from input $x = (u,j) \in \nn \times \bits^{\lceil \log n\rceil}$ computes:
$$
g_s(u,j)\ :=\ G(u) \oplus s_j\ .
$$

\paragraph {(6)} {\bf Claim:} {\em If $s$ satisfies 
$\mm = \bigcup_j s_j \oplus A$ then generator $g_s$ is hard for all Cook-Reckhow proof 
systems.}

\medskip

For the sake of a contradiction assume that a strong proof system $Q$ satisfies
$\bs_Q(\tau(g_s)_b) \le m^{O(1)}$, for some $b \in \mm$ and $n \gg 0$. By Claim (4)
there is $a \in A$ such that $b = a \oplus s_j$ for some $j \le n$.
Hence the $Q$-proof of $\tau(g_s)_b$ can be extended by a size $m^{O(1)}$ proof to
a derivation of $G(x) \neq a$ which is just $\tau(G)_a$. Thus
$\bs_Q(\tau(G)_a) \le m^{O(1)}$ and hence 
$\bs_P(\tau(G)_a) \le m^{O(1)} < t(n)$ too. But that contradicts the definition of $A$.

\qed

The obvious advantage of the demi-bit conjecture 
over Conjecture \ref{conj} is that it 
relates to pseudo-randomness, cryptography and natural proofs, notions that are more
familiar to complexity theorists than proof complexity is. That does not mean that
it is more plausible. 
I perceive as a significant difference between the demi-bit conjecture (and also 
the (E-NP) hypothesis) and Conjecture \ref{conj}
that the later limits the ability of a 
{\em uniform} adversary (an $\np$ set in the complement of the generator) 
while the former limits {\em non-uniform} adversaries (deterministic
or non-deterministic). Considering the fact that we know essentially nothing about 
the power of large unrestricted circuits I think that no hypothesis limiting their 
power ought to be a priori branded as {\em plausible}. Perhaps the qualification {\em consistent}
(tacitly {\em with present knowledge}) would be more accurate (and it has the advantage that 
two contradictory hypotheses we cannot decide at present can be both consistent).

\section{The failure of the fdp} \label{thm}

We would like to combine Theorems \ref{enp} and \ref{ila} but the 
former theorem requires a uniform generator while the generator provided by
the latter theorem is not uniform.
To overcome this obstacle we shall use the gadget generator of 
\cite{Kra-generator} as it can, in a sense, hide the non-uniformity. 
It is defined as follows. Let 
$$
f\ : \ \bits^\ell \times \kk\ \rightarrow\ \bits^{k+1}
$$
be a p-time function where $\ell = \ell(k)$ depends on $k$; we call $f$ the 
{\bf gadget function}.
The {\bf gadget generator based on $f$} is a function
$$
\gadf\ : \ \NN  \rightarrow \NNN
$$
where $N := \ell + k (\ell +1)$ defined as follows:

\begin{enumerate}
		
\item Input $x \in \nn$ is interpreted as $\ell + 2$ strings
$$
v, u^1, \dots, u^{\ell + 1}
$$
where $v \in \bits^\ell$ and $u^i \in \kk$ for all $i$.
	
\item Output $y = \gadf(x)$ is the
concatenation of $\ell + 1$ strings $w^i \in \bits^{k+1}$ where:
$$
w^i\ :=\ f(v, u^i)\ .
$$
\end{enumerate}

For a fixed gadget $v := c \in \bits^\ell$ denote by $f_c$ the function $\kk \rightarrow \bits^{k+1}$ computed 
by the gadget function with $c$ substituted for $v$ and note that
then we have:
\begin{equation} \label{obser}
	\tau(\gadf)_b(v/c) \ = \ \bigvee_{i \in [\ell+1]} \tau(f_c)_{b^i} 
\end{equation}
with no atoms occurring in more than one formula $\tau(f_c)_{b^i}$.

\begin{theorem} \label{main}
Assume the (E-NP) hypothesis and Conjecture \ref{rud}. 
Then there exists a pps $Q$ such that no pps $P$ that simulates
$Q$ has the strong fdp property.
\end{theorem}  

\prf

\paragraph {(1)} Let $g$ be a $\pp/poly$ generator provided by Theorem \ref{ila}
and assume $g_n$ is computed by circuit $C_n$ that is described by $\ell \le n^{O(1)}$ 
bits.

\paragraph {(2)} We take $\gadf$ determined by the following data:
\begin{itemize}

\item $k := \nl$

\item $|v| = \ell$ and $v$ is interpreted as a description of circuit $C$
with $k$ inputs

\item $f(v,u) := C(u)$

\end{itemize}
We have $\gadf : \NN \rightarrow \NNN$ and it is a uniform 
p-time function.

\paragraph {(3)} {\bf Claim} (\cite[L.4.2.6]{k4}): {\em 
Assume (E-NP). Then there is $c \geq 1$ and a map 
$$
H : \bits^{c\log N} \rightarrow \NNN
$$
computed in time $N^{O(1)}$ such that 
$Rng(H) \not \subseteq Rng(\gadf)$.}

\smallskip

$H$ is the Nisan-Wigderson generator (cf. \cite{NW})
based on a function so hard that
$H$ is pseudorandom w.r.t. $Rng(\gadf)$.

\paragraph {(4)} Define a strong proof system $Q$ that extends EF by an extra set
of tautologies
\begin{equation} \label{extra}
\bigvee_{b \in Rng(H)} \tau(\gadf)_b
\end{equation}
for all $n \geq 1$. Note that these extra tautologies
have size $N^{O(1)}$ and that they 
form a p-time set and hence $Q$ is indeed a Cook-Reckhow proof system.

\paragraph {(5)} Assume now that $P$ is any pps that simulates $Q$. Hence 
all tautologies (\ref{extra}) have p-size $P$-proofs. Substitute in them for the gadget
$v$ the (description of) circuit $C_n$. Using observation (\ref{obser}) 
this substitution of (\ref{extra}) will become
$$
\bigvee_{b \in Rng(H)} \bigvee_{i \in [\ell +1]} \tau(g)_{b^i}
$$
where $b = (b^1, \dots, b^{\ell +1}) \in \NNN$.

If $P$ had the strong fdp then for some $b$ and $i$:
$$
\bs_P(\tau(g)_{b^i}) \le N^{O(1)} \le m^{O(1)}\ .
$$
But that contradicts that $g$ is hard for all pps and, in particular, for $P$.

\qed

\section{A remark on uniformity} \label{unif}

Generator $g$ used in the proof of Theorem \ref{main} is not uniform while $\gadf$
used there too is uniform but we do not known if it is also hard. Hence we could not apply
directly Theorem \ref{enp} and we had to bypass this problem by a modified argument.

One way to deduce the hardness of $\gadf$ would be to prove that $g$ is $\bigvee$-hard,
a notion defined for this purpose in \cite{Kra-strong}: by
\cite[Thm.4.1]{Kra-strong} the $\bigvee$-hardness of $g$ would imply 
the hardness of $\gadf$. 
Generator $g$ is {\bf $\bigvee$-hard} for $P$ 
iff for any $c \geq 1$ only finitely many disjunctions of the form
\begin{equation} \label{vhard}
\bigvee_{b \in B} \tau(g_n)_{b} \ ,
\end{equation}
with $B \subseteq \mm$ have a $P$-proof of size at most $m^c$.

Ren et al. \cite{RWZ} showed that (their version of) the generator $g_P$ hard for $P$
has a weaker version of the property of being pseudo-surjective for $P$ 
and pseudo-surjectivity implies $\bigvee$-hardness. 
We will not define the pseudo-surjectivity
here (cf. \cite{Kra-dual}) but just state that their result implies
that no disjunction (\ref{vhard}) has a short $P$-proof under the an additional assumption
that $B$ is small. In order to use the $\bigvee$-hardness 
of $g$ to establish the hardness of $\gadf$
via \cite[Thm.4.1]{Kra-strong} one needs to allow $|B| = \ell + 1$, 
$\ell$ being the size of the gadget. Unfortunately this
falls well outside the bounds allowed in \cite{RWZ}. To determine whether
the generator $g$ is $\bigvee$-hard remains an open problem.

\bigskip
\noindent
{\large {\bf Acknowledgments:}}

I thank O.~Je\v zil, G.~Krej\v c\'{\i} and M.~Otrubov\' a (all in Prague)
for pointing out some misprints, to E.~Khaniki (Oxford) for updating some  
references and to J.~Pich (Oxford) for reminding me of my own observation in past
regarding the truth-table function and feasible interpolation.

\end{document}